%% file: main.tex
\documentclass[conference]{IEEEtran}

\usepackage[usenames, dvipsnames]{color}
\usepackage[cmex10]{amsmath} 
\usepackage{amssymb,amsthm}

\usepackage[nolist]{acronym}

\usepackage{enumerate}

\usepackage{subcaption}
\usepackage{calc}
\usepackage{cite}
\usepackage{graphicx}

\usepackage{array,multirow}

\usepackage{cite}
\usepackage{tikz,pgfplots}

\newtheorem{theorem}{Theorem}

\usepackage{amsmath} 

\input{notation}

\input{acronym}

\newcommand\copyrighttext{%
  \footnotesize \textcopyright~\the\year~IEEE. Personal use of this material is permitted.
  Permission from IEEE must be obtained for all other uses, in any current or future
  media, including reprinting/republishing this material for advertising or promotional
  purposes, creating new collective works, for resale or redistribution to servers or
  lists, or reuse of any copyrighted component of this work in other works.}
\newcommand\copyrightnotice{%
\begin{tikzpicture}[remember picture,overlay]
\node[anchor=north,yshift=-15pt] at (current page.north) {\fbox{\parbox{\dimexpr\textwidth-\fboxsep-\fboxrule\relax}{\copyrighttext}}};
\end{tikzpicture}
\vspace{-0.3cm}
}

\title{
Analysis of Tree-Algorithms with \\ Multi-Packet Reception
}

\author{\IEEEauthorblockN{\v Cedomir Stefanovi\' c\IEEEauthorrefmark{1}, H. Murat G\"ursu\IEEEauthorrefmark{2}, Yash Deshpande\IEEEauthorrefmark{2},  Wolfgang Kellerer\IEEEauthorrefmark{2} \\
\IEEEauthorblockA{
\IEEEauthorrefmark{1}Department of Electronic Systems, Aalborg University, Denmark}}
\IEEEauthorblockA{
\IEEEauthorrefmark{2}Chair of Communication Networks, Technical University of Munich, Germany \\
Email: cs@es.aau.dk, \{murat.guersu,yash.deshpande,wolfgang.kellerer\}@tum.de }}

\begin{document}

\maketitle
\copyrightnotice 
\thispagestyle{empty}
\pagestyle{empty}

\begin{abstract}

In this paper, we analyze binary-tree algorithms in a setup in which the receiver can perform multi-packet reception (MPR) of up to and including $\mpr$ packets simultaneously. The analysis addresses both traffic-independent performance as well as performance under Poisson arrivals.
For the former case, we show that the throughput, when normalized with respect to the assumed linear increase in resources required to achieve $\mpr$-MPR capability, tends to the same value that holds for the single-reception setup.
However, when coupled with Poisson arrivals in the windowed access scheme, the normalized throughput \emph{increases} with $\mpr$, and we present evidence that it asymptotically tends to 1. 
We also provide performance results for the modified tree algorithm with $\mpr$-MPR in the clipped access scheme.
To the best of our knowledge, this is the first paper that provides an analytical treatment and a number of fundamental insights in the performance of tree-algorithms with MPR.

\end{abstract}


\input{introduction}

\input{background}

\input{model}

\input{analysis}

\input{arrivals}

\input{extensions}

\input{conclusions}

\input{appendix}



\addtolength{\textheight}{-12cm}   





\bibliographystyle{IEEEtran}

\bibliography{bibliography}

\end{document}

%% file: notation.tex
\newcommand{\mpr}{K}

\newcommand{\cri}{l}
\newcommand{\pgf}{Q}

\newcommand{\nuser}{n}

\newcommand{\ecri}{L}

\newcommand{\T}{\mathrm{T}}

\newcommand{\E}{\mathrm{E}}

%% file: acronym.tex
\begin{acronym}
\acro{AP}{access point}
\acro{RA}{random access}
\acro{LT}{Luby Transform}
\acro{BP}{belief propagation}
\acro{i.i.d.}{independent and identically distributed}
\acro{IoT}{Internet of Things}
\acro{MTC}{machine type communication}
\acro{PER}{packet error rate}
\acro{SIC}{successive interference cancellation}
\acro{URLLC}{ultra-reliable low latency communication}
\acro{PMF}{probability mass function}
\acro{CRI}{collision resolution interval}
\acro{PGF}{probability generating function}
\acro{CPGF}{conditional probability generating function}
\acro{rhs}{right-hand side}
\acro{CSA}{coded slotted ALOHA}
\acro{MPR}{multi-packet reception}
\acro{MTA}{modified binary-tree algorithm}
\acro{SICTA}{binary-tree algorithm with SIC}
\acro{BTA}{binary tree-algorithm}
\acro{FCFS}{First Come First Served}
\acro{r.v.}{random variable}
\acro{w.r.t.}{with respect to}
\acro{RA}{random access}
\acro{TA}{tree algorithm} 
\end{acronym}

%% file: introduction.tex
\section{Introduction}
\label{sec:Intro}

The promise of \ac{IoT} 
has profoundly influenced research and development of cellular access systems in the past decade.
The recent introduction of massive IoT (mIoT) and ultra-reliable service categories in 5G standardization is likely the best example of the foreseen significance of IoT communications. 
An outstanding feature of IoT communications is the need to efficiently deal with the randomness in users' activations in the access network. 
In mIoT terms, with potentially thousands of devices in a cell~\cite{D2.1}, the vital requirement from the system perspective is the throughput, i.e., the efficiency of the use of system resources.
From users' perspective, the vital requirement is a low outage probability, which is directly related to a bounded access delay.
A combination of these two requirements can be made in the form of stable throughput, which can be defined as the throughput for which the access delay is bounded.  

The proverbial examples of \ac{RA} algorithms are ALOHA~\cite{Abramson:ALOHA} and tree algorithms~\cite{capetanakis1979tree}.
The former serve as the basis for access networking in numerous cellular technologies, due to their simplicity.
However, the latter achieve significantly higher stable throughputs~\cite{massey1981collision}.
Still, the throughput performance in both cases is modest: the maximum stable throughput of slotted ALOHA is $0.3679 \, \frac{\mathrm{packet}}{\mathrm{slot}}$~\cite{Abramson:ALOHA} and for tree-algorithms $0.4878 \, \frac{\mathrm{packet}}{\mathrm{slot}}$~\cite{V1986}, which motivated development of advanced \ac{RA} schemes.
A way to improve the performance is to employ \ac{MPR} that can be achieved through use of CDMA, multi-access coding techniques, MIMO, etc., as it more effectively deals with the mutual interference of concurrently active users.  
\ac{MPR} pushes the performance of ALOHA-based schemes significantly, e.g., see~\cite{NMT2005,GSP2015,BR2018}.  
However, to the best of our knowledge, use of \ac{MPR} in tree-algorithms has not been genuinely investigated and the related work on the topic is scarce.
The work in \cite{gau2011tree} analyzes \ac{MPR} in a tree-algorithm with continuous arrivals with a small number of users in the system, proposing a transmission strategy that guarantees stability.
These aspects limit reaching general conclusions.
The work in~\cite{GSP2018} proposes a hybrid scheme combining user identity decoding via the $\mpr$-\ac{MPR} tree-algorithm and user data decoding via polling.
The performance figures in~\cite{GSP2018} are derived for the scheme as a whole, while the analysis of the tree-algorithm based component is basic and not corresponding to the usual tree-algorithm setup. 

In this paper we establish a communication-theoretic analysis of the use of \ac{MPR} in tree-algorithms and provide some fundamental insights.
For this purpose, we investigate binary tree-algorithms on a $\mpr$-collision channel (as introduced in Section~\ref{sec:model}), which is a simple, but useful abstraction of the $K$-\ac{MPR} capability.
The main contributions are:
\begin{itemize}
\item We investigate the traffic-independent performance of the \ac{BTA}~\cite{capetanakis1979tree} and derive both recursive and non-recursive expressions for the expected length of \ac{CRI}, conditioned on the number of users $\nuser$. 
We compute the throughput of the scheme conditioned on $\nuser$, providing evidence that its value, when normalized with $\mpr$, tends to stabilize around the  value  characteristic for the single-packet reception case.  
\item We analyse the scheme under Poisson arrivals in the windowed access setup and derive bounds on the stable throughput.
We show that the performance improves with $\mpr$, implying that \ac{MPR} does pay-off in windowed access.
\item Finally, we discuss extensions of the analysis, by showing results for the \ac{MTA}~\cite{massey1981collision} and its combination with the clipped access~\cite{RS1990} that is known for its superior performance.
We also discuss applications to other types of tree algorithms~\cite{MF1985}.
\end{itemize}

The paper is organized as follows.
Section~\ref{sec:background} introduces the background and Section~\ref{sec:model} 
the system model.
Section~\ref{sec:BTA_analysis} analyzes the traffic-independent performance of \ac{BTA} on $\mpr$-collision channel.
Section~\ref{sec:arrivals} investigates the performance under Poisson arrivals in windowed access.
Section~\ref{sec:extensions} elaborates extensions of the analysis.
Section~\ref{sec:conclusions} concludes the paper.

%% file: background.tex
\section{Background}
\label{sec:background}


\subsection{Tree Algorithms}
\label{sec:TA}

We are interested in the classical random-access problem in which a set of of randomly arrived users, whose identities and the number $\nuser$ are a-priori unknown, contend for the access to a common \ac{AP}.
We focus on tree algorithms, whose core concept is the collision resolution driven by the feedback from the \ac{AP}.
The basic form of the algorithm, \ac{BTA}, on the standard collision channel (as defined in Section~\ref{sec:CC}), operates as follows.
The channel resources are grouped in slots, and all $\nuser$ users transmit their packets in the first slot.
If the slot is singleton (i.e., $\nuser = 1$), the packet in it is decoded.
In case of a collision (i.e., $\nuser > 1 $), the collision resolution starts.
The users split into two groups, e.g., group 0 and group 1; every user decides which group to join randomly and independently of any other user.
All users in group 0 transmit in the next slot.
If the slot is idle (i.e., no user selected group 0), the users from group 1 transmit in the next slot.
If the slot is singleton, the packet in it becomes decoded and the users from group 1 transmit in the next slot.
Finally, if the slot contains a collision, the users in group 0 split in two new groups and the procedure is recursively repeated until all packets from users in group 0 become decoded, after which the users in group 1 transmit in the next slot.
The procedure is then repeated for the users in group 1. 
The scheme ends when all $\nuser$ packets are decoded.
Should a user transmit in a slot, perform a split, or wait, is decided by monitoring the feedback sent from the \ac{AP}, see Section~\ref{sec:model} for details.
Fig.~\ref{fig:BTA}-a) shows an example of the scheme.
When the scheme is used in a gated access setup, its stable throughput reaches $0.346 \, \frac{\mathrm{packet}}{\mathrm{slot}}$~\cite{capetanakis1979tree}.\footnote{See Section~\ref{sec:arrivals} for details on gated access and stable throughput.}

The work in~\cite{capetanakis1979tree} was followed by a number of improvements and generalizations; we outline only the ones relevant to this work.
In the \ac{MTA}, if the slot dedicated to group 0 in a generic split happens to be idle, the users from group 1 split immediately, thus potentially avoiding repetition of the same collision.
In the example in Fig.~\ref{fig:BTA}-a), there would be an immediate split after slot 9 in the \ac{MTA}.
This simple modification increases the maximum stable throughput in gated access to $0.375\, \frac{\mathrm{packet}}{\mathrm{slot}}$~\cite{massey1981collision}.
Another improvement is to use tree-algorithms in the windowed access setup (see Section~\ref{sec:arrivals}) that pushes the maximum stable throughput, e.g., to 0.4294~$\frac{\mathrm{packet}}{\mathrm{slot}}$ in case of \ac{BTA}~\cite{massey1981collision}.
The performance can be further pushed by using \ac{MTA} with biased splitting in the clipped access setup (which is a modification of the windowed access), achieving the stable throughput of $0.4877 \, \frac{\mathrm{packet}}{\mathrm{slot}}$~\cite{RS1990}.
Finally, the best performing tree-algorithm so far is a variant of the previous scheme, elaborated in~\cite{V1986}, achieving the stable throughput of $0.4878 \, \frac{\mathrm{packet}}{\mathrm{slot}}$.

\begin{figure}[t]
  \centering
  \includegraphics[width=0.88\columnwidth]{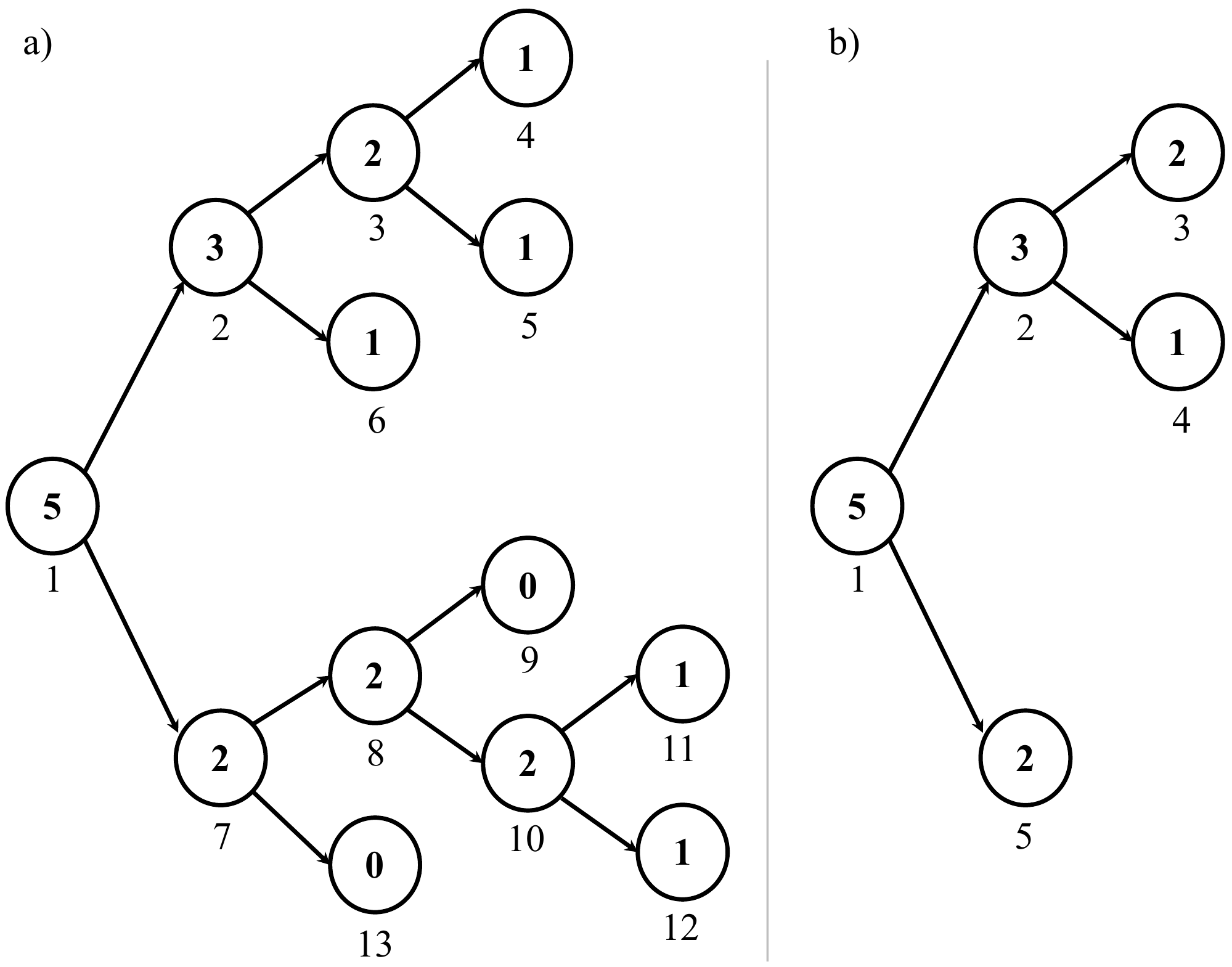}
  \caption{Binary-tree algorithm (BTA): A node represents a slot, the number inside the node shows the number of users (i.e., packets) in the slot, and the number beneath shows the sequence number of the slot. There are $\nuser=5$ users. a) \ac{BTA} on standard, 1-collision channel.  b) \ac{BTA} on 2-collision channel.}
  \label{fig:BTA}
\end{figure}


\subsection{$\mpr$-Collision Channel}
\label{sec:CC}

The performance figures stated in Section~\ref{sec:TA} were derived for the collision channel, the default channel model for the assessment of \ac{RA} algorithms.
Here, we introduce its generalization in the form of $\mpr$-collision channel, as follows:
\begin{enumerate}[(i)]
\item If there are no packets in the slot, the slot is idle. 
\item If there are up to and including $\mpr$ packets colliding in a slot, all packets are successfully decoded.
\item If the number of colliding packets in a slot is above $\mpr$, the slot is a collision slot and no packet can be decoded. 
\end{enumerate}
This model reduces to the standard collision channel for $\mpr=1$.
It is a simple model of \ac{MPR} capability, adequate for $\mpr$-out-of-$\nuser$ coding schemes~\cite{MACbook,OP2017,GSP2018}, or spread-spectrum systems~\cite{MGA2017}.
A feature common to all these works is that the amount of (time-frequency) resources needed to achieve $\mpr$-\ac{MPR} linearly increases with $\mpr$.
That is, slots in $\mpr$-collision channel are $\mpr$ times larger compared to the standard collision channel.
This assumption holds in this paper.
Fig.~\ref{fig:BTA}-b) shows an example of \ac{BTA} on 2-collision channel.

%% file: model.tex
\section{System Model}
\label{sec:model}

We assume a batch of $\nuser$ users, contending for the access to a common \ac{AP} over a $\mpr$-collision channel.\footnote{In this section and in Section~\ref{sec:BTA_analysis}, our focus is on a single batch of $\nuser$ arrived users. Section~\ref{sec:arrivals} deals with the incorporation of an arrival process.}
The channel resources are divided in slots, and users are slot-synchronized.
The users contend by transmitting packets that perfectly fit in slots.
There is also a perfect feedback channel from the \ac{AP}.

All $\nuser$ users transmit in the first slot that appears on the channel. 
The interval elapsed from the first slot up to and including the slot when all $\nuser$ users become resolved is denoted as \ac{CRI}.
The length of a \ac{CRI} conditioned on $\nuser$ is a \ac{r.v.} denoted by $\cri_\nuser$.

\begin{table}[t]
  \begin{center}
    \caption{Values of feedback and states of the counters of the users contending in the example in Fig.~\ref{fig:BTA}-b).}
    \label{tab:BTA_example}
    \begin{tabular}{|c|c|c|c|c|c|c|}
    \hline
      \multirow{2}{*}{\textbf{slot no.}} & \multicolumn{5}{c| } {\textbf{state of counter} } & \multirow{2}{*}{\textbf{feedback}} \\ \cline {2-6}
                & \textbf{\textit{user 1}} & \textbf{\textit{user 2}} & \textbf{\textit{user 3}} & \textbf{\textit{user 4}} & \textbf{\textit{user 5}} & \\ \hline
      1         & 0     & 0     & 0     & 0     & 0         & e \\ \hline
      2         & 0     & 1     & 0     & 0     & 1         & e \\ \hline
      3         & 0     & 2     & 1     & 0     & 2         & 1 \\ \hline
      4         & -1    & 1     & 0     & -1    & 1         & 1 \\ \hline
      5         & /     & 0     & -1    & /     & 0         & 1 \\ \hline
      end       & /     & -1    & /     & /     & -1        & / \\ \hline
    \end{tabular}
  \end{center}
\end{table}

We now proceed by elaboration of the. \ac{BTA}.
We denote the slots in a \ac{CRI} by their numbers, starting with $1$.
The number of users transmitting in slot $j$ is denoted by $n_j$.
After every slot $j$, the \ac{AP} broadcasts the following feedback
\begin{align}
f_j = \begin{cases}
 0, &\text{ if } n_j = 0 \\
 1, &\text{ if } n_j \leq \mpr \\
 \text{e}, &\text{ if } n_j > \mpr .
\end{cases}  
\end{align}
Each user $i$ maintains a counter, whose state in slot $j$ is denoted by $C_{i,j}$, where $C_{i,1}=0$, $\forall i$, i.e., the initial value is 0.
The state of the counter determines the user's transmission decisions:
(i) if $C_{i,j}=0$, user $i$ transmits in slot $j$, (ii) if $C_{i,j} > 0$, user $i$ abstains from transmitting in slot $j$, and (iii) if $C_{i,j} < 0$, user $i$ has become resolved and does not contend further.
The state of the counter is updated through feedback and splitting  procedure, as follows:
\begin{align}
\label{eq:counter}
C_{i,j + 1} = \begin{cases}
 b_{i,j}, &\text{ if } f_j = e \text{ and } C_{i,j} = 0 \\
 C_{i,j} + 1, &\text{ if } f_j = e \text{ and } C_{i,j} > 0 \\
 C_{i,j} - 1, &\text{ if } f_j \in \{0,1\}
\end{cases}  
\end{align}
where $b_{i,j}$ is a Bernoulli \ac{r.v.} with $\Pr \{ b_{i,j} = 0 \} = p$ and $\Pr \{ b_{i,j} = 1 \} = 1-p$ (if $p=\frac{1}{2}$, the splitting is fair; otherwise, the splitting is biased).
The topmost case in \eqref{eq:counter} pertains to a split after the collision in slot $j$, where $b_{i,j}$ determines whether user $i$ will join the generic group 0 or group 1.
To facilitate easier understanding, Table~\ref{tab:BTA_example} lists the values of the feedback and the states of the users' counters for the example in Fig.~\ref{fig:BTA}-b), assuming that: in slot 2, user 1, 3, and 4 chose group 0, while user 2 and 5 chose group 1; and in slot 3, user 1 and 4 chose group 0 and user 3 chose group 1.

\subsection*{Performance Parameters}

The primary parameter of interest is the expected \ac{CRI} length conditioned on $\nuser$, denoted by $\ecri_\nuser = \mathrm{E}[ \cri_\nuser ]$.
We are also interested in the conditional throughput, defined as
\begin{align}
\label{eq:T}
    \T_n = \frac{1}{\mpr}\frac{n}{\ecri_\nuser}.
\end{align}
The throughput in \eqref{eq:T} is the measure of the efficiency of resource use, where the normalization with $\mpr$ reflects the increase in the amount of resources contained in a slot required to achieve $\mpr$-\ac{MPR}.

%% file: analysis.tex
\section{Analysis of Traffic-Independent Performance}
\label{sec:BTA_analysis}

The length of a \ac{CRI} conditioned on $\nuser$ is
\begin{align}
    \cri_\nuser = \begin{cases}
    1, & n \leq \mpr \\
    1 + \cri_{i} + \cri_{\nuser - i}, & n > \mpr
    \end{cases}
\end{align}
where $i$ is the number of users that joined group 0 and $\nuser-i$ is the number of users that joined group 1.

The expected conditional length of \ac{CRI} for $n>K$ is simply
\begin{align}
\label{eq:BTA_ecri}
    \ecri_\nuser = 1 + \sum_{i = 0}^\nuser {\nuser \choose i} p^{ i } ( 1 - p )^{\nuser -i} (\ecri_i  + \ecri_{\nuser - i }) .
\end{align}
By developing \eqref{eq:BTA_ecri}, $\ecri_\nuser$ can be  calculated recursively through
\begin{align}
\label{eq:BTA_ecri_1}
    \ecri_\nuser = \begin{cases} 
    1, & \nuser \leq \mpr  \\
    \frac{1 + \sum_{i = 0}^{K} g (\nuser,i,p)}{ 1 - g (\nuser,0,p)} & \nuser = \mpr + 1\\
    \frac{1 + \sum_{i = 0}^{K} g (\nuser,i,p)  + \sum_{i = K + 1}^{\nuser - 1} g (\nuser,i,p) \ecri_i}{ 1 - g (\nuser,0,p)} & \nuser > \mpr + 1
    \end{cases}
\end{align}
where $g(n,i,p)={n \choose i}\left( p^i ( 1- p )^{n-i} + p^{n-i}(1 - p)^i \right)$.

The next theorem gives the non-recursive expression for $\ecri_\nuser$.
\begin{theorem}
\label{th:non_recursive}
\begin{align}
\label{eq:BTA_ecri_final}
\ecri_\nuser = 1 -  {\nuser \choose \mpr} \sum_{j = 1}^{\nuser - \mpr}  \frac{2 \, j \, (-1)^{j} {\nuser -\mpr \choose j}} {(j + \mpr) ( 1 - p^{j+\mpr}-(1-p)^{j+\mpr})}.
\end{align}
\end{theorem}
\begin{proof}
See appendix.
\end{proof}
\begin{figure}[t]
\centering
\input{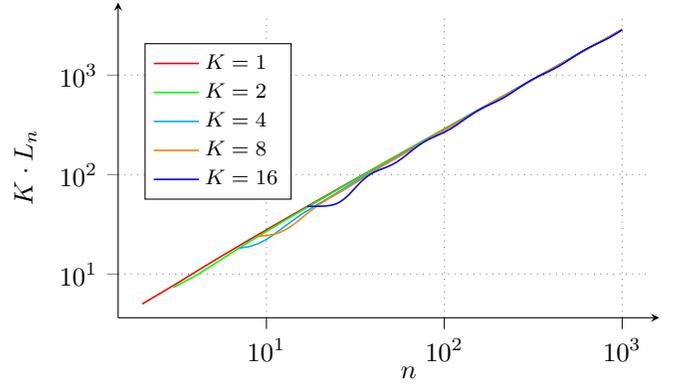}
\caption{Conditional length of \ac{CRI} $\ecri_\nuser$ of \ac{BTA} on $\mpr$-collision channel for varying number of  users $\nuser$ and $\mpr$.} \label{fig:BTA_L}
\end{figure}

The value of \eqref{eq:BTA_ecri_final} is minimized for $p=\frac{1}{2}$, as the value of the denominator on \ac{rhs} becomes maximized.
In other words, fair splitting achieves minimal $\ecri_\nuser$, given by
\begin{align}
\label{eq:BTA_ecri_fair}
\ecri_\nuser = 1 -  {\nuser \choose \mpr} \sum_{j = 1}^{\nuser - \mpr}  \frac{2 \, j \, (-1)^{j} {\nuser -\mpr \choose j}} {(j + \mpr) ( 1 - 2^{-j-\mpr +1})}.
\end{align}       
In the rest of the paper, we assume fair splitting. 

Fig.~\ref{fig:BTA_L} shows $\mpr \cdot \ecri_\nuser$ as function of $\nuser$, for $\mpr\in \{1,2,4,8,16\}$; the multiplication with $\mpr$ makes the comparison fair.
Obviously, as $\nuser$ grows, the difference among curves vanish, and a linear trend of increase in $\mpr \cdot \ecri_\nuser$ emerges (a fact exploited in Section~\ref{sec:arrivals} when developing bounds for $\ecri_\nuser$).

Fig.~\ref{fig:BTA_T_Alt} plots the values of the conditional throughput $\T_n$ corresponding to $\ecri_\nuser$ in Fig.~\ref{fig:BTA_L}.
The maxima $\T_\nuser = 1$ are reached for $\nuser = K$, when the \ac{CRI} lasts a single slot.
As $\nuser$ increases, $\T_\nuser$ shows a damped oscillatory behaviour, where the oscillations' amplitude increases with $\mpr$. 
Also, $\T_\nuser$ tends to stabilize around $0.346$, irrespective of the value of $\mpr$; this fact was already identified for $\mpr = 1$~\cite{massey1981collision}. 
The conclusion is that, as $\nuser$ grows, increase in $\mpr$-\ac{MPR} capability does not increase the efficiency of the use of the system resources.
A potential way to formally prove the conclusion is to extend the asymptotic analysis from~\cite{MF1985}.
However, the related investigation is beyond the scope of the paper and is part of our ongoing work.

%% file: arrivals.tex
\section{Incorporation of Traffic Arrivals}
\label{sec:arrivals}

So far, we analyzed properties of BTA on a $\mpr$-collision channel assuming that the number of users $\nuser$ is given.
An obvious way to incorporate traffic arrivals, where $\nuser$ is a \ac{r.v.}, is via gated access.
In this case, all users that arrive during a \ac{CRI} wait until that \ac{CRI} ends, 
and then transmit in the next available slot, initiating the next \ac{CRI}.
An alternative is to use windowed access, where the time-axis of arrivals is divided into equal-length windows and every window is ``served'' in a separate \ac{CRI}, as follows: all users arriving in the $i$-th window transmit in the first slot after the \ac{CRI} of the users arriving in the $(i-1)$-th window ends, thus starting their own \ac{CRI}.


\begin{figure}[t]
\centering
\input{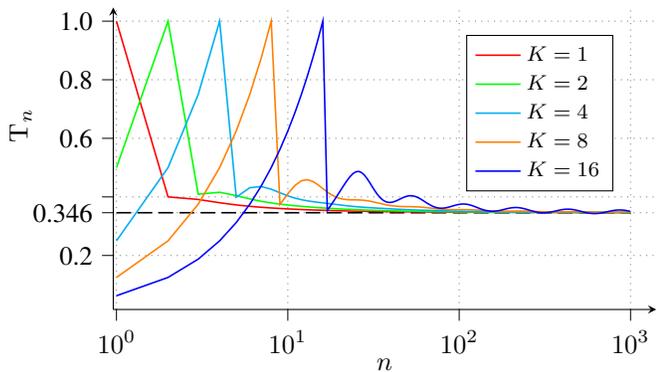}
\caption{Conditional throughput $\T_\nuser$ of \ac{BTA} on $\mpr$-collision channel for varying number of arrived users $\nuser$ and $\mpr$.} \label{fig:BTA_T_Alt}
\end{figure}

\ac{BTA} with windowed access significantly outperforms the gated access version in terms of maximum stable throughput for $\mpr=1$, as outlined in Section~\ref{sec:TA}.
For Poisson arrivals, which are of our interest, the maximum stable throughput is defined as the maximum arrival rate per slot $\lambda$ with a finite expected delay of user resolution.
A fairly tight upper bound on maximum stable throughput for \ac{BTA} with gated access and $\mpr=1$ is $\lambda < 0.346 \, \frac{\mathrm{packet}}{\mathrm{slot}}$~\cite{massey1981collision}, which bears direct relation to the asymptotic behaviour of $\T_n$, see Fig.~\ref{fig:BTA_T_Alt}.
The figure also suggests that the maximum stable throughput will not increase with $\mpr$, as the asymptotic behaviour of $\T_n$ does not change with $\mpr$. 
Thus, hereafter 
we focus on performance evaluation of \ac{BTA} on the $\mpr$-collision channel with windowed access. 

\begin{table}[t]
  \begin{center}
    \caption{\ac{BTA} on $K$-collision channel with windowed access: Bounds on  throughput and optimal window sizes.}
    \label{tab:WA_example}
    \begin{tabular}{|c|c|c|c|c|c|c|c|}
       \hline
       $\mpr$ & $\alpha_m$ & $\beta_m$ & $ \lambda_\text{U} / \mpr$ & $\lambda_\text{S} / \mpr $  & $\lambda_\text{S} \Delta_\text{S} $ & $\Delta_\text{S}$ \\ \hline 
       1 &  2.88538 & 2.8854 &  0.42951  & 0.42951  & 1.149 &  2.675 \\ \hline
       2 & 1.44267 & 1.44272 & 0.47068 & 0.47068  & 1.831 &  1.945 \\ \hline
       4 & 0.72158 & 0.72116 & 0.51751 & 0.51751  & 3.2 &  1.546 \\ \hline
       8 & 0.35907 & 0.36214 & 0.56779 & 0.56779 & 5.967 &  1.314 \\ \hline
       16 & 0.17355 & 0.1859 & 0.62388 & 0.62388 & 11.753 &  1.177 \\ \hline
    \end{tabular}
      \end{center}
\end{table}

Assume a generic arrival window and denote its length by $\Delta$ slots (not necessarily an integer).
The probability of $n$ arrivals ($\nuser \in \mathbb{N}$) in the window is
$\Pr\{ N = \nuser \} = \frac{\lambda \Delta }{\nuser!}e^{ - \lambda \Delta}$,    
i.e., $n$ is a Poisson \ac{r.v.} with mean $\lambda \Delta$.
The expected length of \ac{CRI} is
\begin{align}
    \ecri(\lambda\Delta) = \E \{ \ecri_\nuser | \lambda \Delta \} = \sum_{\nuser=0}^\infty   \ecri_\nuser \frac{(\lambda\Delta)^\nuser}{\nuser!} e^{-\lambda\Delta}
\label{eq:windowed_Acc}
\end{align}
and the goal is to find $\lambda\Delta$ for which
\begin{align}
\label{eq:WA_stability}
 \ecri(\lambda\Delta) < \Delta.  
\end{align}
This condition is necessary for stability, ensuring that arrivals in the window will be served on average in an interval shorter than the window.\footnote{For stability to hold, the condition $\E \{ \ecri^2_\nuser \} < \infty$ has also to be satisfied. This is indeed the case if $\ecri(\lambda\Delta) < \Delta$, but we omit the proof.} 
One can compute~\eqref{eq:windowed_Acc} with arbitrary precision via \eqref{eq:BTA_ecri_final} or \eqref{eq:BTA_ecri_fair} for any $\lambda$ and $\Delta$, and establish the bound in \eqref{eq:WA_stability}.
Hereafter, we proceed with a convenient and insightful bounding method from~\cite{massey1981collision}.
It can be shown that
\begin{align}
\label{eq:BTA_bound}
    \alpha_m n - 1 \leq \ecri_n \leq \beta_m n - 1
\end{align}
which holds for any $m \geq 1$ and $n > m$,
where $\alpha_m = \inf_{n>m} \frac{\sum_{i=0}^{m - 1}{n \choose i}(\ecri_i + 1)}{\sum_{i=0}^{m - 1}{n \choose i}i}$ and $\beta_m = \sup_{n>m} \frac{\sum_{i=0}^{m - 1}{n \choose i}(\ecri_i + 1)}{\sum_{i=0}^{m - 1}{n \choose i}i}$.
By substituting \eqref{eq:BTA_bound} into \eqref{eq:windowed_Acc}, it can be shown that
\begin{align}
    f (\alpha_m, m, \lambda \Delta) \leq L ( \lambda \Delta ) \leq f (\beta_m, m, \lambda \Delta) 
\end{align}
where
\begin{align}
    f ( x, k, z ) = x \cdot z - 1 + \sum_{i = 0}^{k } ( L_i - x \cdot i + 1 ) \frac{ z^i }{i !} e^{ -z}.
\end{align}
The scheme will be stable if
\begin{align}
      f (\beta_m, m, \lambda\Delta) < \Delta 
\end{align}
which yields the bound
\begin{align}
    \label{eq:WA_SB}
    \lambda < \sup_{\lambda \Delta > 0 } \frac{\lambda \Delta}{f (\beta_m, m, \lambda\Delta)} = \lambda_\text{S}.
\end{align}
Similarly, the bound for which the scheme will be unstable is
\begin{align}
\label{eq:WA_UB}
    \lambda > \sup_{\lambda \Delta > 0 } \frac{\lambda \Delta}{f (\alpha_m, m, \lambda\Delta)} = \lambda_\text{U}.
\end{align}

Table~\ref{tab:WA_example} lists values of $\alpha_m$, $\beta_m$, $\frac{\lambda_\text{U}}{\mpr}$, $\frac{\lambda_\text{S}}{\mpr}$, $\lambda_\text{S} \Delta_\text{s}$ and $\Delta_\text{S}$ (the value of $\Delta$ that corresponds to $\lambda_\text{S}$ in \eqref{eq:WA_SB}),
obtained for $m = 50$.
Obviously, the bounds $\frac{\lambda_\text{S}}{\mpr}$ and $\frac{\lambda_\text{U}}{\mpr}$ very tight, their difference can be observed only for decimal places not shown in the table, where $\frac{\lambda_\text{S}}{\mpr} < \frac{\lambda_\text{U}}{\mpr}$.
Notably, $\lambda_\text{S} / \mpr$ \emph{increases} with $\mpr$, implying that use of \ac{MPR} can indeed pay off.
Another interesting observation is that the values of $\lambda_\text{S} \Delta_\text{s}$ are less than $\mpr$ for $\mpr > 1$.
An intuitive explanation for this can be found in Fig.~\ref{fig:BTA_T_Alt}.
The maximum value of $\T_\nuser$ is 1 for any $\mpr$, attained for $\nuser = \mpr$, and, for $\mpr > 1$, the values of $\T_\nuser$ are higher for $\nuser$ that approaches $\mpr$ from the left, than for $\nuser > \mpr$.
Thus, the optimal\footnote{Optimal in the sense that it maximizes the stable throughput.} mean of the Poisson distribution of the number of arrivals in a window is less than $\mpr$. 
Also, it can be observed that the optimal window size in slots $\Delta_\text{s}$ decreases with $\mpr$.

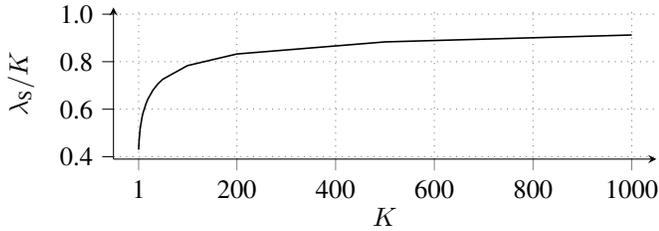
\begin{figure}[t]
\centering
\input{asymptotics_windowed}
\vspace{-0.3cm}
\caption{Maximum stable throughput $\lambda_\text{S}/K$ for increasing $\mpr$.} \label{fig:windowed_asymp}
\end{figure}

Fig.~\ref{fig:windowed_asymp} plots $\lambda_\text{S}/\mpr$ as function of $\mpr$.
Obviously, the bound on maximum stable throughput for this scheme steadily increases with $\mpr$.
We conjecture that $\lim_{\mpr \rightarrow \infty} \lambda_\text{S}/\mpr = 1$, leaving the proof for our future work.





%% file: asymptotics_windowed.tex
\begin{tikzpicture}


\begin{axis}[
     axis lines=left,
     width=8.8cm,
     height=3.6cm,
     grid=both,
     grid style={dotted,gray},
     legend cell align=left,
    legend style={at={(0.55,0.75)},anchor=west},
tick align=outside,
tick pos=both,
xlabel={$\mpr$},
xmajorgrids,
xmin=-50, xmax=1050,
xtick={1,200,400,600,800,1000},
xticklabels={1,200,400,600,800,1000},
ylabel={$\lambda_{\text{S}}/\mpr$},
ymajorgrids,
ymin=0.39, ymax=1.0285245,
ytick style={color=black},
ytick={0.4,0.6,0.8,1},
yticklabels={0.4,0.6,0.8,1.0}
]

\addplot [semithick, black]
table {windowed_lambda.tex};

\end{axis}

\end{tikzpicture}

%% file: extensions.tex
\section{Extensions to Other Variants of Tree-Algorithms}
\label{sec:extensions}


A straightforward extension of the analysis can be made for the case of \ac{MTA} on a $\mpr$-collision channel, omitted here due to space constraints.
Instead, we depict $\T_\nuser$ of the scheme with fair splitting in Fig.~\ref{fig:MTA_T}, where the y-axis is zoomed to emphasize the most important details.
Obviously, as $\nuser$ grows, $\T_\nuser$ drops as $\mpr$ increases, and tends to the value characteristic for \ac{BTA}.
This is a consequence of the fact that a collision slot involves at least $\mpr$ packets, and with increasing $\mpr$, the probability of observing an idle slot immediately after a fair split of such slot decreases, counteracting the premise from which \ac{MTA} draws its gain with respect to \ac{BTA}.

We proceed by evaluating the performance of \ac{MTA} with fair splitting and clipped access\footnote{See~\cite{RS1990} for the details of this scheme.} on a $\mpr$-collision channel.
The impact of $\mpr$-\ac{MPR} can be taken into account by adjusting the expected number of resolved users per \ac{CRI} in the analysis presented in~\cite{RS1990}; we omit the details due to space constraints.
The evolution of $\lambda_\text{S} / \mpr$ for this scheme with $\mpr$ is given in Table~\ref{tab:FCFS_kmpr}.
Obviously, when compared with $\lambda_\text{S} / \mpr$ in Table~\ref{tab:WA_example}, the performance gap between \ac{BTA} with windowed access and \ac{MTA} with clipped access diminishes as $\mpr$ grows, due to the reasons outlined in the previous paragraph. 

\begin{figure}[t]
\centering
\input{MTA_Oscillations}
\caption{Conditional throughput $\T_\nuser$ of \ac{MTA} on $\mpr$-collision channel for varying number of arrived users $\nuser$ and $\mpr$.} \label{fig:MTA_T}
\end{figure}
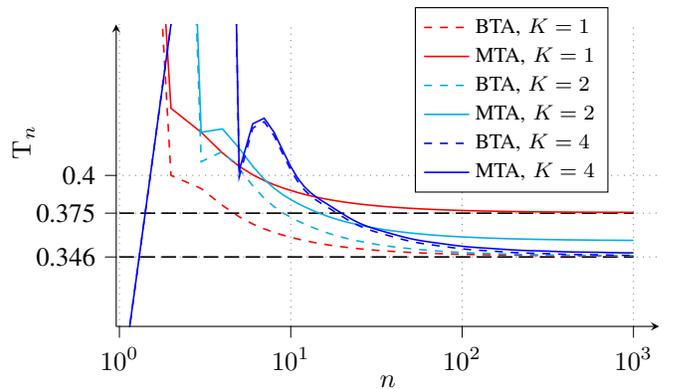

The analysis can be also extended to $d$-ary tree algorithms, where we expect that use of $\mpr$-\ac{MPR} will bring some new insights.
Specifically, for $\mpr = 1$, the ternary \ac{MTA} with biased splitting features the superior performance~\cite{MF1985}.
However, the analogous conclusion for $\mpr > 1 $ is not obvious, as increase in $\mpr$ increases the chances that packets contained in a slot become decoded, an effect also achieved with $d$-ary splitting.
The corresponding investigation is part of our ongoing work.




%% file: MTA_Oscillations.tex
\begin{tikzpicture}


\begin{axis}[
     axis lines=left,
     width=8.8cm,
     height=5.6cm,
     grid=both,
     grid style={dotted,gray},
     legend cell align=left,
    legend style={at={(0.55,0.75)},anchor=west},
log basis x={10},
tick align=outside,
tick pos=both,
xlabel={$\nuser$},
xmajorgrids,
xmin=0.957945784384138, xmax=1412.53754462275,
xmode=log,
xtick style={color=black},
xtick={0.01,0.1,1,10,100,1000,10000,100000},
xticklabels={\(\displaystyle {10^{-2}}\),\(\displaystyle {10^{-1}}\),\(\displaystyle {10^{0}}\),\(\displaystyle {10^{1}}\),\(\displaystyle {10^{2}}\),\(\displaystyle {10^{3}}\),\(\displaystyle {10^{4}}\),\(\displaystyle {10^{5}}\)},
ylabel={\; \; \; \; $\T_\nuser $},
ymajorgrids,
ymin=0.3, ymax=0.5,
ytick style={color=black},
ytick={0,0.2,0.346,0.375, 0.4,0.6,0.8,1,1.2},
yticklabels={0.0,0.2,0.346,0.375, 0.4,0.6,0.8,1.0,1.2}
]
\path [draw=black, semithick, dash pattern=on 5.55pt off 2.4pt]
(axis cs:1,0.346)
--(axis cs:1000,0.346);

\path [draw=black, semithick, dash pattern=on 5.55pt off 2.4pt]
(axis cs:1,0.375)
--(axis cs:1000,0.375);

\addplot [semithick,dashed, red]
table {BTA_kmpr_1.txt};
\addlegendentry{\footnotesize{BTA, $\mpr= 1$}}

\addplot [semithick, red]
table [col sep=comma]{mta_mpr1.txt};
\addlegendentry{\footnotesize{MTA, $\mpr= 1$}}

\addplot [semithick,dashed, cyan]
table {BTA_kmpr_2.txt};
\addlegendentry{\footnotesize{BTA, $\mpr= 2$}}

\addplot [semithick, cyan]
table [col sep=comma] {mta_mpr2.txt};
\addlegendentry{\footnotesize{MTA, $\mpr= 2$}}

\addplot [semithick,dashed, blue]
table {BTA_kmpr_4.txt};
\addlegendentry{\footnotesize{BTA, $\mpr= 4$}}

\addplot [semithick, blue]
table [col sep=comma] {mta_mpr4.txt};
\addlegendentry{\footnotesize{MTA, $\mpr= 4$}}
\end{axis}

\end{tikzpicture}

%% file: conclusions.tex
\section{Conclusions}

\begin{table}[t]
  \begin{center}
    \caption{Maximum stable throughput for \ac{MTA} with clipped access on $\mpr$-collision channel.}
    \label{tab:FCFS_kmpr}
    \begin{tabular}{|c|c|c|c|c|c|c|c|}
       \hline
       $\mpr$ & $1$ & $2$ & $ 4$ & $8 $ & $16$ \\ \hline
      $\lambda_{\text{S}}/ \mpr$ &   0.48703 &  0.4923  &   0.52257   &   0.56844 &   0.62388 \\ \hline
    \end{tabular}
      \end{center}
\end{table}

\label{sec:conclusions}
In this paper we analyzed binary tree algorithms with \ac{MPR}, which is seen as a promising communication technique for massive access scenarios in 5G.
Our results show that adoption of \ac{MPR} in \ac{BTA} with windowed access has a beneficial effect on the stable throughput performance.
Interestingly, as \ac{MPR}-capability increases, the performance gap between \ac{MTA} with clipped access and \ac{BTA} with windowed access closes, implying that the added complexity of the former scheme does not pay-off.

%% file: appendix.tex
\appendix
\label{sec:appendix}

\begin{proof}[Proof of Theorem~\ref{th:non_recursive}]
We prove the theorem using the method elaborated in~\cite{MF1985}.
We first compute the conditional \ac{PGF} of $\cri_\nuser$, given by
\begin{align}
\pgf_\nuser(z) = \mathrm{E} \left\{ z^{\cri_n} \right\}, \text{ where } \pgf_\nuser (z) = z \text{ for } 0 \leq \nuser \leq \mpr .
\label{eq:BTA_ic}
\end{align}
%
It is straightforward to verify that, for $\nuser > \mpr$, 
\begin{align}
\label{eq:BTA_cpgf}
\pgf_\nuser(z)  = z \sum_{i=0}^\nuser  { \nuser \choose i } p^{ i} ( 1 - p )^{\nuser - i}  \pgf_{i}(z) \pgf_{\nuser - i}(z).
\end{align}

We proceed by computing (unconditional) \ac{PGF} of \ac{CRI}, assuming that 
$\nuser$ is a Poisson \ac{r.v.} with mean $x$.
The \ac{PGF} is
\begin{align}
& \pgf (x,z)  =  \sum_{\nuser = 0}^\infty \pgf_\nuser (z) \frac{x^\nuser}{\nuser!} e^{-x} \label{jpgf} \\
& = z\sum_{\nuser = 0}^\infty \frac{x^\nuser}{\nuser!} e^{-x} \sum_{i=0}^{\nuser} { \nuser \choose i } p^{ i} ( 1 - p )^{\nuser -i}  \pgf_{i}(z)\pgf_{\nuser - i}(z) \, +  \nonumber \\ 
        & + \left(z - z^3\right) \sum_{k=0}^{\mpr} \frac{x^k}{k!} e^{-x} \label{upgf}
\end{align}
where we exploited
\eqref{eq:BTA_ic} and \eqref{eq:BTA_cpgf}, and used the fact that 
\begin{align}
    \sum^\nuser_{i=0} {n \choose i} p^{i} (1 - p)^{\nuser - i} \pgf_{i}(z) \pgf_{\nuser - i}(z)  = z^2,  \;  n \leq \mpr.
\end{align}
With some manipulation, the first term on \ac{rhs} in \eqref{upgf} becomes
\begin{align}
z \sum^\infty_{i=0} \pgf_{i}(z) \frac{\left( p x \right)^{i}}{i!} e^{-p x} \sum^\infty_{n = 0} & \pgf_{n}(z) \frac{\left( ( 1 - p ) x \right)^{n}}{n!} e^{-(1 - p ) x}.
\end{align}
Thus, \eqref{upgf} can be re-written as
\begin{align}
\label{upgf2}
\pgf(x,z) = z\pgf(p x,z)\pgf\left( (1-p) x,z \right)  + \left(z-z^{3}\right)\sum_{k=0}^{\mpr} \frac{x^k}{k!} e^{-x} .
\end{align}
To calculate the first moment of the \ac{PGF}, we derive \eqref{upgf2} \ac{w.r.t.} $z$ and substitute $z=1$.
This yields the mean of the (unconditional) \ac{CRI} length, denoted by $\ecri (x) $   
\begin{align}
\ecri (x) = 1 + \ecri (p x) +\ecri ((1-p) x) - 2\sum_{k=0}^{\mpr} \frac{x^k}{k!}  e^{-x}
\label{eq:step_2_n}
\end{align}
where we used the fact that $\pgf(x, 1) = 1$, $\forall x$.

Similarly, deriving \eqref{jpgf} \ac{w.r.t.} $z$ and substituting $z=1$ yields
\begin{align}
\label{ps_ecri}
L(x) = \sum^\infty_{n=0} \ecri_\nuser \frac{x^\nuser}{\nuser!}  e^{-x} \stackrel{ \text{(a)}}{=} \sum^\infty_{\nuser=0} a_\nuser x^\nuser
\end{align}
where (a) stems from an assumed power series representation of $\ecri (x)$.
By multiplying \eqref{ps_ecri} with $e^{x}$, and then applying the power series expansion of $e^x$, it can be shown that
\begin{align}
\ecri_\nuser  = \sum^\nuser_{j=0} \frac{n!}{(\nuser-j)!}a_j.
\label{non_recursive}
\end{align}
We now find $a_j$, $j= 0,1,...,n$.
From \eqref{eq:BTA_ecri_1}, it follows that
\begin{align}
    a_j = \begin{cases}
    1, & j = 0 \\
    0, & j = 1,\dots, \mpr .
    \end{cases}
    \label{alpha_weight_t}
\end{align}
Further, substituting \eqref{ps_ecri} into \eqref{eq:step_2_n} and using the power series expansion of $e^{-x}$ yields
\begin{align}
 \sum^\infty_{\nuser=0}  a_\nuser  ( 1 - p^\nuser & - (1 - p)^\nuser) x^\nuser   = 1 - 2 \sum_{k=0}^{\mpr} \frac{x^k}{k!}\sum^\infty_{\nuser=0} (-1)^\nuser \frac{x^\nuser}{\nuser!} \nonumber \\
 & = 1 + 2\sum_{\nuser = 0}^\infty \sum_{k=0}^{\min(\nuser,\mpr)} \frac{( - 1)^{\nuser-k+1}}{ k! (\nuser - k )!} x^\nuser.
 \label{eq:ps}  
\end{align}
For $\nuser \leq K$, it can be shown that  
\begin{align}
\sum_{k=0}^{\nuser} \frac{( - 1)^{\nuser-k+1}}{ k! (\nuser - k )!} =
\begin{cases}
-1, & n = 0 \\
0, & 0 < n \leq \mpr
\end{cases}
\end{align}
which, coupled with \eqref{alpha_weight_t}, transforms \eqref{eq:ps} into
\begin{align}
 \sum^\infty_{\nuser=\mpr + 1} a_\nuser ( 1 - p^\nuser - & (1 - p)^\nuser) x^\nuser \nonumber = \\ 
 & = 2\sum_{\nuser = \mpr + 1}^\infty \sum_{k=0}^\mpr \frac{( - 1)^{\nuser-k+1}}{ k! (\nuser - k )!} x^\nuser .  \label{eq:exp_coeffs}
\end{align}
Solving \eqref{eq:exp_coeffs} for $a_\nuser$, $\nuser \geq \mpr$, we get
\begin{align}
a_\nuser  = {\sum_{k=0}^{\mpr} \frac{(-1)^{\nuser-k+1}}{k! (\nuser - k)!}}\cdot \frac{2}{1 - p^\nuser-(1-p)^\nuser}.
\label{alpha_weight_n}
\end{align}
By substituting \eqref{alpha_weight_t} and \eqref{alpha_weight_n} in \eqref{non_recursive} for $\nuser > \mpr$, and after some manipulation, we get
\begin{align}
\ecri_\nuser = 1 + \sum_{j = \mpr + 1}^\nuser {\nuser \choose j}  \frac{2 (-1)^{j+1}}{1 - p^j-(1-p)^j}\sum_{k=0}^{\mpr} {j \choose k} (-1)^{-k}.
\label{eq:mpr_bias_wosic}
\end{align}
Using the identity that holds for $ \mpr < j$
\begin{align}
    \sum_{k=0}^{\mpr} (-1)^{k} {j \choose k}  = ( - 1)^\mpr { j - 1 \choose \mpr }
\end{align}
we further simplify \eqref{eq:mpr_bias_wosic} to 
\begin{align}
& \ecri_\nuser  = 1 + \sum_{j = \mpr + 1}^\nuser {\nuser \choose j}   {j-1 \choose \mpr} \frac{2\, (-1)^{j -\mpr + 1} }{1 - p^j - ( 1 - p )^ j }. 
\label{eq:mpr_bias_wosic_2}
\end{align}
Finally, we get
\begin{align}
\ecri_\nuser = 1 -  {\nuser \choose \mpr} \sum_{j = 1}^{\nuser - \mpr}  \frac{2 \, j \, (-1)^{j} {\nuser -\mpr \choose j}} {(j + \mpr) ( 1 - p^{j+\mpr}-(1-p)^{j+\mpr})}
\end{align}
which concludes the proof.
\end{proof}

%% file: main.bbl
\begin{thebibliography}{10}
\providecommand{\url}[1]{#1}
\csname url@samestyle\endcsname
\providecommand{\newblock}{\relax}
\providecommand{\bibinfo}[2]{#2}
\providecommand{\BIBentrySTDinterwordspacing}{\spaceskip=0pt\relax}
\providecommand{\BIBentryALTinterwordstretchfactor}{4}
\providecommand{\BIBentryALTinterwordspacing}{\spaceskip=\fontdimen2\font plus
\BIBentryALTinterwordstretchfactor\fontdimen3\font minus
  \fontdimen4\font\relax}
\providecommand{\BIBforeignlanguage}[2]{{%
\expandafter\ifx\csname l@#1\endcsname\relax
\typeout{** WARNING: IEEEtran.bst: No hyphenation pattern has been}%
\typeout{** loaded for the language `#1'. Using the pattern for}%
\typeout{** the default language instead.}%
\else
\language=\csname l@#1\endcsname
\fi
#2}}
\providecommand{\BIBdecl}{\relax}
\BIBdecl

\bibitem{D2.1}
\BIBentryALTinterwordspacing
{E2E-aware Optimizations and advancements for Network Edge of 5G New Radio
  (ONE5G)}, ``{Deliverable D2.1: Scenarios, KPIs, use cases and baseline system
  evaluation},'' Tech. Rep., nov 2017. [Online]. Available:
  \url{{https://one5g.eu/}}
\BIBentrySTDinterwordspacing

\bibitem{Abramson:ALOHA}
N.~Abramson, ``The {ALOHA} system -- {A}nother alternative for computer
  communications,'' in \emph{Proc. of 1970 Fall Joint Computer Conf.},
  vol.~37.\hskip 1em plus 0.5em minus 0.4em\relax AFIPS Press, 1970, pp.
  281--285.

\bibitem{capetanakis1979tree}
J.~Capetanakis, ``Tree algorithms for packet broadcast channels,'' \emph{IEEE
  transactions on information theory}, vol.~25, no.~5, pp. 505--515, 1979.

\bibitem{massey1981collision}
J.~L. Massey, ``Collision-resolution algorithms and random-access
  communications,'' in \emph{Multi-user communication systems}.\hskip 1em plus
  0.5em minus 0.4em\relax Springer, 1981, pp. 73--137.

\bibitem{V1986}
S.~Verdu, ``Computation of the efficiency of the {M}osely-{H}umblet contention
  resolution algorithm: {A} simple method,'' \emph{Proceedings of the IEEE},
  vol.~74, no.~4, pp. 613--614, Apr. 1986.

\bibitem{NMT2005}
V.~{Naware}, G.~{Mergen}, and L.~{Tong}, ``Stability and delay of finite-user
  slotted {ALOHA} with multipacket reception,'' \emph{IEEE Trans. Info.
  Theory}, vol.~51, no.~7, pp. 2636--2656, Jul. 2005.

\bibitem{GSP2015}
J.~Goseling, C.~Stefanovic, and P.~Popovski, ``{A} pseudo-{B}ayesian approach
  to sign-compute-resolve slotted {ALOHA},'' in \emph{IEEE International
  Conference on Communications (ICC) 2015, MASSAP Workshop}, London, UK, Jun.
  2015.

\bibitem{BR2018}
A.~{Baiocchi} and F.~{Ricciato}, ``Analysis of pure and slotted {ALOHA} with
  multi-packet reception and variable packet size,'' \emph{{IEEE} Commun.
  Lett.}, vol.~22, no.~7, pp. 1482--1485, Jul. 2018.

\bibitem{gau2011tree}
R.-H. Gau, ``Tree/stack splitting with remainder for distributed wireless
  medium access control with multipacket reception,'' \emph{IEEE transactions
  on wireless communications}, vol.~10, no.~11, pp. 3909--3923, 2011.

\bibitem{GSP2018}
J.~{Goseling}, C.~{Stefanovic}, and P.~{Popovski}, ``Sign-compute-resolve for
  tree splitting random access,'' \emph{{IEEE} Trans. Inf. Theory}, vol.~64,
  no.~7, pp. 5261--5276, Jul. 2018.

\bibitem{RS1990}
R.~Rom and M.~Sidi, \emph{Multiple Access Protocols}.\hskip 1em plus 0.5em
  minus 0.4em\relax Springer, 1990, ch.~5.

\bibitem{MF1985}
P.~Mathys and P.~Flajolet, ``{Q}-ary collision resolution algorithms in
  random-access systems with free or blocked channel access,'' \emph{{IEEE}
  Trans. Inf. Theory}, vol.~31, no.~2, pp. 217--243, Mar. 1985.

\bibitem{MACbook}
D.~Danyev, B.~Laczay, and M.~Ruszinko,
  ``\BIBforeignlanguage{English}{{M}ultiple access adder channel},'' in
  \emph{\BIBforeignlanguage{English}{Multiple Access Channels}}, E.~Biglieri
  and L.~Gyorfi, Eds.\hskip 1em plus 0.5em minus 0.4em\relax IOS press, 2007,
  pp. 26--53.

\bibitem{OP2017}
O.~Ordentlich and Y.~Polyanskiy, ``Low complexity schemes for the random access
  {G}aussian channel,'' in \emph{Proc. 2017 IEEE Int. Symp. Inf. Theory},
  Aachen, Germany, Jun. 2017, pp. 2533--2537.

\bibitem{MGA2017}
A.~{Mengali}, R.~{De Gaudenzi}, and P.~{Arapoglou}, ``Enhancing the physical
  layer of contention resolution diversity slotted {ALOHA},'' \emph{{IEEE}
  Trans. Commun.}, vol.~65, no.~10, pp. 4295--4308, Oct. 2017.

\end{thebibliography}
